\documentclass[aoas,preprint]{imsart}

\RequirePackage{amsthm,amsmath,amsfonts,amssymb}
\RequirePackage[authoryear]{natbib}
\RequirePackage[colorlinks,citecolor=blue,urlcolor=blue]{hyperref}
\RequirePackage{graphicx}

\RequirePackage[boxed]{algorithm2e}
\RequirePackage{float}
\RequirePackage{soul}

\usepackage{booktabs}
\usepackage[table,xcdraw]{xcolor}

\startlocaldefs

\newcommand{\Rlogo}{\protect\includegraphics[height=11pt,keepaspectratio]{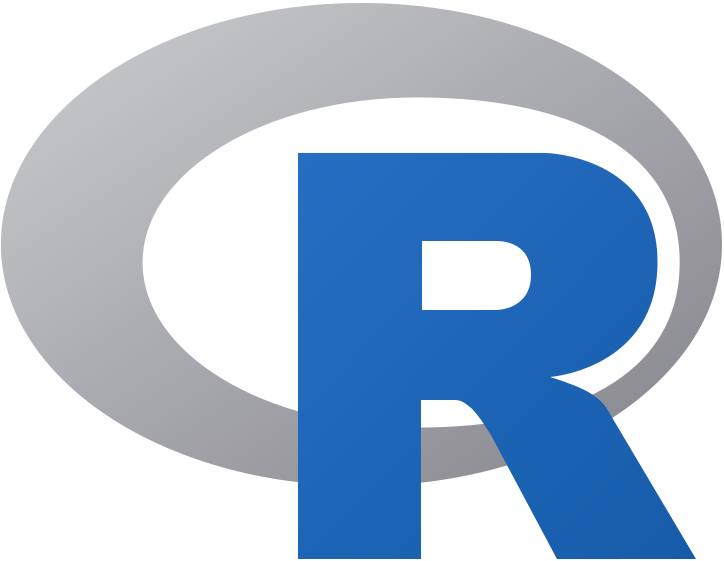}}

\theoremstyle{plain}
\newtheorem{theorem}{Theorem}[section]

\theoremstyle{remark}

\endlocaldefs

\begin{document}

\begin{frontmatter}
\title{Bayesian classification, anomaly detection, and survival analysis using network inputs with application to the microbiome}
\runtitle{Bayesian models using network inputs}

\begin{aug}
\author[A]{\fnms{Nathaniel} \snm{Josephs}},
\author[B]{\fnms{Lizhen} \snm{Lin}},
\author[A]{\fnms{Steven} \snm{Rosenberg}} \\
\and
\author[A]{\fnms{Eric D.} \snm{Kolaczyk}}

\runauthor{Josephs et al.}

\address[A]{Department of Mathematics and Statistics, Boston University}
\address[B]{Department of Applied and Computational Mathematics and Statistics, The University of Notre Dame}
\end{aug}

\begin{abstract}
While the study of a single network is well-established, technological advances now allow for the collection of multiple networks with relative ease.
Increasingly, anywhere from several to thousands of networks can be created from brain imaging, gene co-expression data, or microbiome measurements.
And these networks, in turn, are being looked to as potentially powerful features to be used in modeling.
However, with networks being non-Euclidean in nature, how best to incorporate them into standard modeling tasks is not obvious.
In this paper, we propose a Bayesian modeling framework that provides a unified approach to binary classification, anomaly detection, and survival analysis with network inputs.
We encode the networks in the kernel of a Gaussian process prior via their pairwise differences and we discuss several choices of provably positive definite kernel that can be plugged into our models.
Although our methods are widely applicable, we are motivated here in particular by microbiome research (where network analysis is emerging as the standard approach for capturing the interconnectedness of microbial taxa across both time and space) and its potential for reducing preterm delivery and improving personalization of prenatal care.
\end{abstract}


\end{frontmatter}


\section{Introduction}

Preterm delivery is a global health problem and hence an important area of research.
According to the World Health Organization,\footnote{https://www.who.int/en/news-room/fact-sheets/detail/preterm-birth} it is the leading cause of death among infants worldwide, but it is estimated that over 75\% of preterm infants would survive with appropriate intervention with a survival rate higher among infants born the latest.
Therefore, it is important to be able to identify women at a high risk of preterm delivery and also be able to predict the gestational length in high-risk pregnancies.
Fundamentally, these questions can be understand as problems of binary classification (is a pregnancy at risk of being preterm?), anomaly detection (can we identify high risk pregnancy when the prevalence is ``only" 5-18\%?), and survival analysis (can we understand gestational length as a time-to event?).
The ability to answer these questions promises to improve the personalization of prenatal care.

Recently, the microbiome has emerged in the study of preterm delivery and one dataset that is uniquely suited to addressing all three questions of classification, anomaly detection, and survival analysis is from \cite{digiulio2015temporal}.
The authors tracked the microbiomes of 40 women over the course of their pregnancies by collecting over 3400 samples from the vagina, distal gut, saliva, and tooth/gum areas.
This effort resulted in a rich dataset that revealed the presence of over 1200 different taxa.
The authors' aim was to identify microbial taxa associated with a higher risk of preterm delivery, which occurred in 15 of the 40 women in the study.
We summarize the lengths of gestation in Figure~\ref{fig: histogram}.

\begin{figure}[!htb]
\centering
\includegraphics[width=\textwidth]{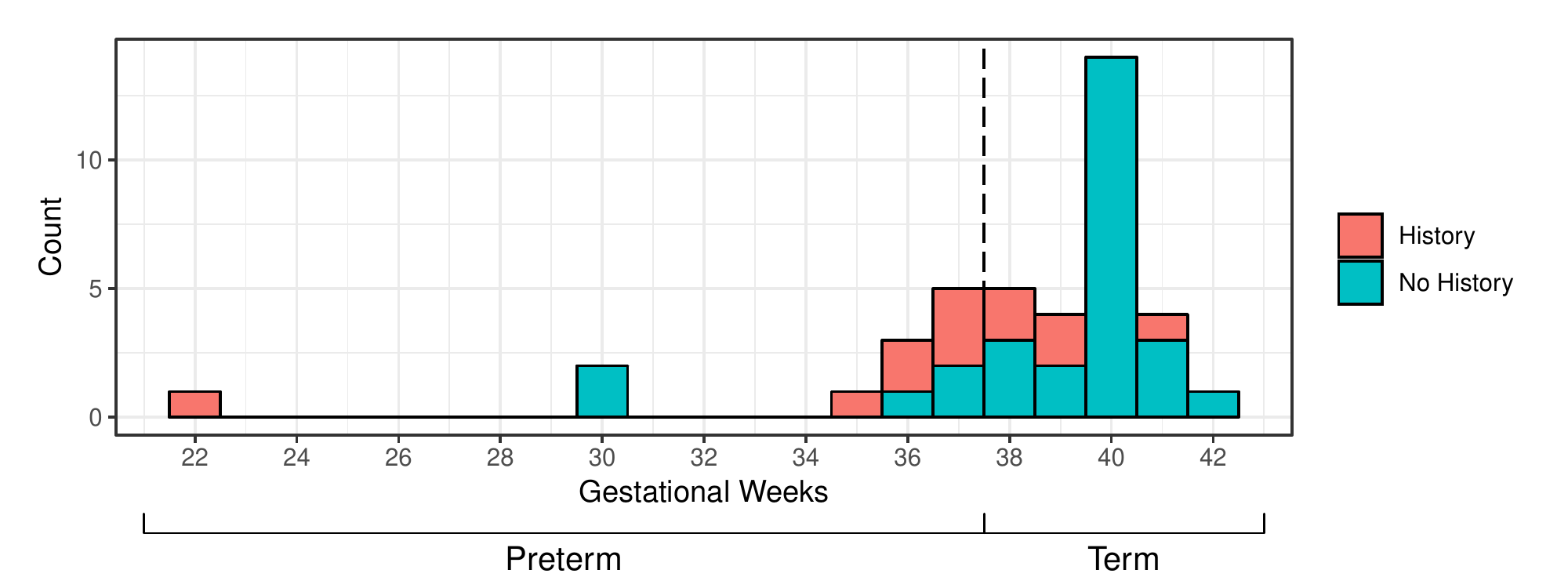}
\caption{Histogram of lengths of gestational periods in weeks. The dashed line separates the classes of term, which is delivery after 37 weeks, from preterm. Red and blue represent women with and without a history of preterm delivery, respectively.}
\label{fig: histogram}
\end{figure}

In their analysis, \cite{digiulio2015temporal} successfully identified several taxa whose dynamics throughout the gestational period were significantly related to a higher risk of preterm delivery.
Additionally, these authors looked at the interconnectedness of taxa at the level of taxonomic communities.
Ultimately, corresponding network summary values were used as descriptors for linear mixed-effect models.
The resulting findings suggest that individual microbiomes -- as captured through network-based representations -- can be important descriptors for various pregnancy outcomes.
For a recent review of the current state of applying network analyses to microbiome data, see \cite{layeghifard2017disentangling}.

Examples of such networks are shown in Figure~\ref{fig: networks}.
In \cite{digiulio2015temporal} and other microbiome analyses, the strategy has been to identify and extract certain potentially relevant aspects of the topology of these networks for downstream analysis.
However, it is clear from these 
visualizations that microbiome networks possess a rich topology and that they can exhibit important topological differences across, for example, patient subgroups.
Without knowing precisely which aspects of network topology might be most relevant to a given study, it is increasingly desirable to have statistical methods that allow researchers to incorporate the \textit{entire} network, rather than just numerical summaries.

\begin{figure}[!htb]
\centering
\includegraphics[width=\textwidth, trim=0in 4in 0in 4in]{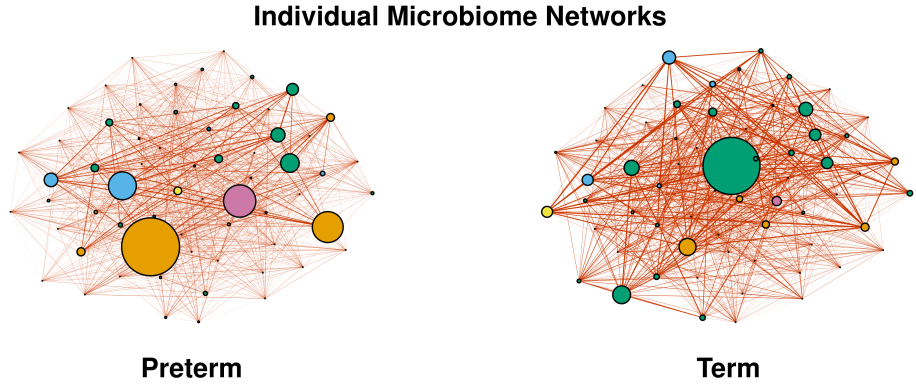}
\caption{Examples of individual microbiome networks. Nodes represent different species of bacteria. They are colored by phylum and their size is proportional to their abundance across the individual's samples. Edges correspond to association between the species with the edge thickness corresponding to the strength of association and color corresponding to the sign (black is negative and red is positive). Further details of the network constructions are outlined in Section~\ref{MB_nets}. Although it appears that the preterm microbiome network has a lower density and is dominated by the presence of a few species, we show in Section~\ref{EDA} that the differences in preterm and term microbiome networks cannot be easily reduced to a few summary statistics.}
\label{fig: networks}
\end{figure}

Herein, we propose a unified methodology for analyzing network datasets for a number of standard purposes using whole-network inputs.
In particular, we develop a classifier that simultaneously serves as the basis for the corresponding anomaly detection and survival analysis methodologies we also develop in this paper in a seamless and integrated manner.
A schematic diagram of the core infrastructure is given in Figure~\ref{fig: schematic}.

We expect our contributions to advance the use of networks to analyze the microbiome, both now and in the future.
Ultimately, a better understanding of the relationship between preterm delivery and microbiome networks that are sufficiently easy to obtain and construct could help to inform routine prenatal care decisions, as well as to provide important controls in studies of pregnant women.
Specifically, our Bayesian approach can help mitigate the challenges of small sample size \citep{mcneish2016using}, which typifies most current microbiome experiments, and at the same time our approach scales to the expected paradigm shift toward microbiome studies that sample with higher resolution for many individuals.

\begin{figure}[!htb]
\centering
\frame{\includegraphics[width=\textwidth]{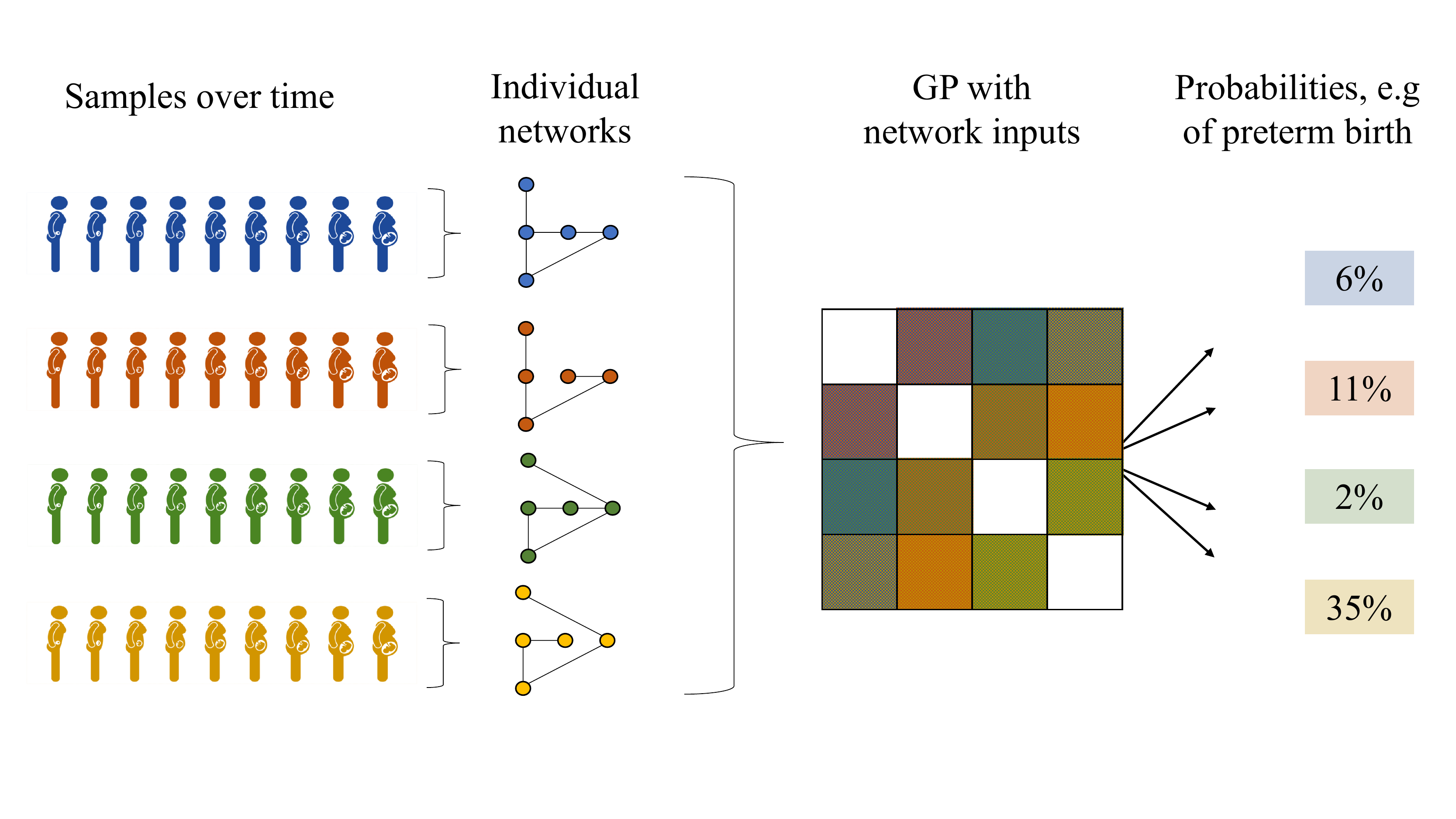}}
\caption{Schematic diagram of study design, network construction, and probabilistic modeling with network inputs. The $4 \times 4$ symmetric matrix $K$ represents a kernel that is used to measure (dis)similarity between pairs of networks. In this toy example, the color entries are pairwise averages of the colors of the corresponding networks, e.g. the purple in entry $K_{1,2}$ is the average of the blue network (row 1) and the orange network (column 2).}
\label{fig: schematic}
\end{figure}

\subsection{Background}

We briefly review relevant network and Gaussian process (GP) concepts.
For more thorough reviews of the statistical analysis of network data and the use of GPs in machine learning, see \cite{kolaczyk2014statistical} and \cite{Rasmussen:2005:GPM:1162254}, respectively.

A \textit{network} or \textit{graph} $G$ is an ordered pair $(V, E)$ consisting of a set of nodes or vertices, $V$, and edges, $E$, where $e \in E(G)$ is a pair of nodes $(v_1, v_2) \in V(G) \times V(G)$.
Throughout, we work with \textit{undirected} and \textit{simple} networks, meaning that there is no directional component to our edges and no self loops.
Furthermore, we will work with networks of the same \textit{order}, i.e. for networks $G_1, \ldots, G_m$, we have $|V_1| = |V_2| = \cdots = |V_m| = n$.

For computations and analyses, graphs are represented through matrices.
The \textit{adjacency matrix} of a graph $G$ with nodes $V(G) = (v_1, \ldots, v_n)$ is an
$n \times n$ matrix $A$ with
\[A_{ij} = \begin{cases}
1 & \text{if } v_i \text{ adjacent to } v_j \\
0 & \text{otherwise} \enskip .
\end{cases}\]
For \textit{weighted networks}, which are triples $G = (V, E, W)$, the corresponding \textit{weighted adjacency matrix} is defined similarly:
\[A_{ij} = \begin{cases}
w_{ij} & \text{if } \{i, j\} \in E \\
0 & \text{otherwise} \enskip . \end{cases}\]

A \textit{Gaussian process} (GP) is a stochastic process that is commonly used for modeling a random function whose evaluation at a finite number of points is jointly Gaussian.
In particular, a GP is often used as a prior distribution for the underlying regression function or latent classification function in a Bayesian nonparametric model which leads to a posterior distribution of the function for inference.
We write
\[f(x) \sim \mathcal{GP}\big(m(x), \ k(x, x')\big) \enskip ,\] 
where
$m(x) = \mathbb{E}[f(x)]$ is the mean function and $k(x, x') = \mathbb{E}\big[\big(f(x) - m(x)\big)\big(f(x') - m(x')\big)\big]$ is the covariance kernel.
Then for any $x_1, \ldots, x_n$, $\left(f(x_1),\ldots, f(x_n)\right)$ is multivariate Gaussian with mean vector $\left(m(x_1),\ldots, m(x_n)\right)$ and covariance matrix $K=((k(x_i,x_j))_{1\leq i\leq j\leq n}$.
One of the most widely used covariance kernels for $x\in \mathbb{R}^d$ is the squared-exponential kernel,
\begin{equation}
    k(x_i, x_j) = \sigma^2\exp\Big(-\frac{|x_i-x_j|^2}{2\ell^2}\Big) \label{eq: K} \enskip ,
\end{equation}
where $|\cdot|$ refers to Euclidean distance and $\sigma^2$ and $\ell$ are the \textit{signal variance} and \textit{length-scale} parameters.

Note that GPs can be defined on any topological space.
In our case, this corresponds to the space of all networks on $n$ nodes.
Therefore, each $x$ corresponds a network input.
We propose to use a GP as a prior distribution for modeling a function on networks which could correspond to a regression function, latent classification map, or the survival curve in a given statistical model with networks given as covariate information.
One of the key difficulties in using a GP model for networks lies in constructing a valid covariance kernel which is a positive definite kernel on the space of networks.
We propose several kernels in Section \ref{Binary Classification} that aim to capture different networks properties. 

The covariance kernel such as the squared-exponential type given in \eqref{eq: K} is equipped with additional hyperparameters that need to be learned.
To do so, we adopt a Bayesian framework by employing hierarchical models for our parameters and hyperparameters.
This has the further benefit of being probabilistic, which, in addition to providing uncertainty quantification of our posterior predictions, yields an immediate extension from our classifier to anomaly detection.
Finally, the Bayesian framework allows for prior specification such as clinicians informing priors for the baseline hazard function in a survival analysis.

\subsection{Related Work}

There is a large literature on graph kernels beginning with \cite{kondor2002diffusion}, who propose diffusion kernels for graphs and suggest that they could be used in conjunction with GPs.
Shortly thereafter, \cite{gartner2003graph} propose the random walk graph kernel, which the authors use to perform a variant of GP regression.
To the best of our knowledge, this is the only explicit use of GPs with graphs.
\cite{kashima2002kernels} apply similar graph kernels for classification, but not using a GP framework.
Since then, there have been many other graph kernels proposed.
For surveys on graph kernels, see \cite{vishwanathan2010graph, kriege2020survey, nikolentzos2019graph}.

Regarding the use of network inputs, a Bayesian approach seems not to have been applied to classification problems.
However, kernel support vector machines (SVMs) have been a popular tool for classification with network inputs and extensions exist to one-class classification and survival analysis. In particular, graph data has been used with SVMs to perform classification of protein function prediction \citep{borgwardt2005protein}, chemical informatics  \citep{ralaivola2005graph}, and disease \citep{rudd2018application}, as well as one-class classification for media data \citep{mygdalis2016graph}.
For the connection between the objective functions of SVMs and GP classification, see \cite[Chapter~6]{Rasmussen:2005:GPM:1162254}.

Alternatively, \cite{relion2019network} provide a frequentist approach to classification with network inputs.
Interestingly, in comparing their method to others, the authors say that ``kernel methods were no better than random guessing."
They also claim that ``kernel methods [are] unsuitable for large scale networks."
Throughout, we demonstrate that, in fact, our method is highly scalable to large networks and outperforms those authors' method on their application and other classification tasks.

Elsewhere, \cite{zhang2013network} introduce Net-Cox for survival analysis with network inputs, which is a network-constrained Cox regression using the graph Laplacian as a penalty.
The authors note that ``surprisingly, network-based survival analysis has not received enough attention."
We will discuss the benefits of our survival analysis methodology, which is not a proportional hazards model, in Section~\ref{Survival Analysis}.

Finally, there is a growing body of work on statistical analysis of multiple networks in general.
For example, Jain and colleagues have a number of contributions, summarized, for example, in \cite{jain2016geometry}, where they propose the use of linear classifiers in graph space.
While our work also has some geometrical underpinnings, it is nevertheless distinct. Similarly, there is various work on network-based extensions of averages \citep{ginestet2017hypothesis,durante2017nonparametric,tang2018connectome,kolaczyk2020averages}, regression \citep{cornea2017regression}, clustering \citep{paperwithpurna}, and PCA \citep{dai2018principal}, to name a few.
Finally, there are several latent space models for multiple networks, for example in \cite{arroyo2019inference, gollini2016joint, salter2017latent}.
Again, however, our work is distinct.

\subsection{Paper Outline}

The remainder of the paper is organized as follows.
In Section~\ref{Models}, we set up our models for binary classification, anomaly detection, and survival analysis.
We discuss implementation and theoretical results for our models in Section~\ref{Implementation} and Section~\ref{Theory}, respectively.
We present the results of our models on simulated networks in Section~\ref{Simulations}, as well our analysis of the microbiome data in Section~\ref{MB}.
Finally, we conclude in Section~\ref{Conclusion} with a discussion of future directions
for this work.

\section{Models}
\label{Models}

\subsection{Binary Classification}
\label{Binary Classification}

Suppose we have data $(G_1, Y_1), \ldots, (G_m, Y_m)$, where
\[G \in \mathcal{G} := \big\{\text{(labeled) weighted networks of order } n\big\}\]
and $Y$ is binary.
We are interested in learning the classification map $\pi(G) := p(y = 1 \ | \ G)$.
To do this, we will directly model $\pi(G)$ using Bayesian inference.
In particular, we equip $\pi(G)$ with a prior that is a deterministic transformation of a GP on
$\mathcal{G}$.
That is, $\pi(G) = H(f(G))$ for some link function $H$ and latent function $f := f(G)$.
Throughout, we work with the logistic link, $\text{logistic}(z) = 1/(1+\exp(-z))$.
We then place a GP prior on $\mathcal{G}$ with kernel $K$ over $f$, where $K$ is given by
\begin{equation}
    K_{ij} = \sigma^2\exp\Big(-\ell \cdot d(G_i, G_j)\Big) \enskip . \label{kernel:sq}
\end{equation}
The squared-exponential kernel in \eqref{kernel:sq} is infinitely differentiable and stationary, which we will use to show posterior consistency in Theorems~\ref{thm: class cons} and \ref{thm: surv cons}.
The challenge in working with this kernel is in specifying an appropriate graph distance.
Although there are many such distances \citep{donnat2018tracking}, a squared-exponential kernel is positive definite \textit{if and only if} its distance is induced by an inner product \citep{jayasumana2013kernel}, in this case on the space of networks $\mathcal{G}$.
In Section~\ref{PD Kernels}, we show that the following distances satisfy this requirement:
\begin{equation}
    d_{\text{F}}\big(G_1, G_2\big) := \frac{1}{n(n-1)}\sum_{i, j \in V} \big(A_{ij}^{(1)} - A_{ij}^{(2)}\big)^2 \enskip , \label{dist:Frobenius}    
\end{equation}
and
\begin{equation}
    d_\lambda\big(G_1, G_2\big) := \sum_{i=0}^{n-1}\big\vert \lambda_i^{(1)} - \lambda_i^{(2)}\big\vert^2 \enskip , \label{dist:lambda}
\end{equation}
where $A^{(k)}$ is the adjacency matrix for network $G_k$ and $\lambda_0^{(k)} \leq \lambda_1^{(k)} \leq \cdots \leq \lambda_{n-1}^{(k)}$ are the eigenvalues of its Laplacian.

The distance in \eqref{dist:Frobenius} is the Frobenius distance, a slight modification of the Hamming distance for binary networks, which captures local measure of difference, and the distance in \eqref{dist:lambda} is a spectral-based distance with a more global measure of difference.

We note that graph-based distances may not always capture the complex topology of the network's distribution.
There are many other graph kernels proposed in the literature, but the fundamental challenge is ensuring its positive definiteness.
One positive definite alternative is the random-walk kernel in \cite{vishwanathan2010graph}, which can be easily plugged in to our method, though it has fewer kernel hyperparameters.
In Section~\ref{Simulations}, we isolate the benefits of our GP approach by comparing a GP classifier with a random walk kernel against an SVM with the same kernel.

For a squared-exponential kernel, our model, including inverse-gamma priors on the kernel hyperparameters and the posterior from which we will make inference on
$\pi(G)$, is specified as follows:
\begin{align*}
    \text{likelihood} &: \hspace{.25cm} p(y \ | \ f) = \prod_{k=1}^m \text{logistic}(f_ky_k) \enskip ;\\
    \text{prior} &: \hspace{.25cm} p(f \ | \ G, \sigma^2, \ell) \propto e^{-\frac{1}{2}f^TK^{-1}f} \enskip ,\\
    &\ \ \ \hspace{.25cm} p(\sigma^2) \propto (\sigma^2)^{-\alpha_1 - 1}e^{-\frac{\beta_1}{\sigma^2}} \enskip ,\\
    &\ \ \ \hspace{.25cm} p(\ell) \propto \ell^{-\alpha_2 - 1}e^{-\frac{\beta_2}{\ell}} \enskip ;\\
    \text{posterior} &: \hspace{.25cm}  p(f, \sigma^2, \ell \ | \ G, y) \propto p(y \ | \ f)p(f \ | \ G, \sigma^2, \ell)p(\sigma^2)p(\ell)\enskip .
\end{align*}
Again, any positive definite kernel may be used in place of $K$ and the only difference would be the prior specification of its hyperparameters (or removal of these priors, for example in the case of the random walk kernel).

Regarding our likelihood, note that if we take $Y \in \{\pm 1\}$, then we can write
\[p(y \ | \ f) = \text{logistic}(f\cdot y) \enskip ,\]
as $\text{logistic}(-z) = 1 - \text{logistic}(z)$ and $p(y = 1 \ | \ G) + p(y = -1 \ | \ G) = 1$.

For a new network $\tilde{G}$, the posterior predictive distributions are
\begin{align*}
    p(\tilde{f} \ | \ G, y, \tilde{G}) &\propto p(\tilde{f} \ | \ G, \tilde{G}, f) p(f \ | \ G, y) \enskip ,\\
    p(\tilde{y} = 1 \ | \ \tilde{G}, G, y) &= \int \text{logistic}(\tilde{f})\cdot p(\tilde{f} \ | \ G, y, \tilde{G}) d\tilde{f} \enskip ,
\end{align*}
where
\[p(\tilde{f} \ | \ G, \tilde{G}, f) \sim N\Big(k(\tilde{G})^TK^{-1}f, k(\tilde{G}) - k(\tilde{G})^TK^{-1}k(\tilde{G})\Big) \enskip .\]

Not surprisingly, our posterior is intractable, which means that our posterior predictive distributions are also intractable.
In Section~\ref{GPC Implementation}, we discuss implementation strategies including Markov chain Monte Carlo (MCMC) sampling and Monte Carlo integration.

\subsection{Anomaly Detection}
\label{Anomaly Detection}

One-class classification (OCC) is a classifier for a single class.
For a review, see \cite{khan2009survey}.
Such a classifier is a type of anomaly detection that is useful when most of the data is from a single class and only a few, if any, training points are from another class.

The most common approach to OCC is through decision boundaries, which are typically found via SVMs.
There are many variations of one-class SVM (OSVM) depending on the availability of training data, i.e whether any training data is from the negative class.
Another approach to OCC is to simply use machine learning techniques for classification, like k-NN or tree-based methods, but these require negative training examples.
However, neither OSVM nor non-OSVM solutions are probabilistic.

GP methods have been proposed as well for the OCC problem.
\cite{kemmler2013one} suggests four GP values for OCC scores:
\begin{align*}
\tilde{\mu} &= \mathbb{E}[\tilde{f} \ | \ G, y, \tilde{G}] \enskip ,  & \tilde{\pi} &= p[\tilde{y} = 1 \ | \ G, y, \tilde{G}] = \text{logistc}(\tilde{\mu}) \enskip , \\
\tilde{\sigma}^2 &=  \text{Var}[\tilde{f} \ | \ G, y, \tilde{G}] \enskip , & \text{H} &= \tilde{\mu} \cdot \tilde{\sigma}^{-1} \enskip .
\end{align*}

These values are obtained immediately from the posterior of our classification model.
Unsurprisingly, \cite{kemmler2013one} notes that ``tuning hyperparameters for one-classification tasks is a difficult task in a general setting without incorporating further model assumptions,'' but we overcome this challenge through our Bayesian approach.

\subsection{Survival Analysis}
\label{Survival Analysis}

\cite{fernandez2016gaussian} introduce a semi-parametric Bayesian method for survival analysis that uses GPs to model variation around a parametric baseline hazard function.
Their method easily incorporates covariates, censoring, and prior knowledge, while avoiding the proportional hazards constraint.
By using our kernel inputs, we are able to adapt this model to perform survival analysis with network inputs.

Suppose we have data $(G_1, T_1), \ldots, (G_m, T_m)$, where $G \in \mathcal{G}$ and $T$ is a survival time on $\mathbb{R}^+$ with survival function $S$ and hazard function $\lambda$.
The authors model $T$ as the first jump of a Poisson process with intensity $\omega$.
Together with the GP priors, the model for $T_i \ | \ G_i$ is
\begin{align*}
    T_i \ | \ \omega_i &\overset{ind}\sim \omega(T_i)e^{-\int_0^{T_i}\omega_i(s)ds} \enskip ,\\
    \omega_i(t) \ | \ f, \omega_0(t), G_i &\overset{\hphantom{ind}}= \omega_0(t)\cdot\text{logistic}(f(t, G_i)) \enskip ,\\
    f(\cdot) &\overset{\hphantom{ind}}\sim \mathcal{GP}(0, K) \enskip ,
\end{align*}
where $K$ is a kernel in network and time.
We take
\begin{equation}
    K\big((t_1, G_1), (t_2, G_2)\big) = K_T(t_1, t_2) + K_G(G_1, G_2) \enskip , \label{eq: K surv}
\end{equation}
where $K_T$ and $K_G$ are the kernels given in is a kernel for time given in \eqref{eq: K} and \eqref{kernel:sq}, respectively, with shared signal variance $\sigma^2$.
The authors prove (\textit{Proposition 1}) that for stationary kernels such as the squared-exponential, $S(t)$ associated with $f(t)$ is a proper survival function, i.e. for a fixed network $G$, we have $S_G(t) = \mathbb{P}(T > t \ | \ G) \to 0$ as $t \to \infty$.

Unfortunately, the likelihood for $T_i \ | \ G_i$ is (doubly) intractable since $\omega_i$ is defined by a GP.
To overcome this, the authors develop a data augmentation scheme for sampling from an inhomogeneous Poisson process with intensity $\omega_0(t)$, which allows a tractable reformulation of the model.
Using the tractable model, which is a joint distribution on $(R, T)$, where $R$ are (unknown) rejected jump points from a thinned Poisson process and $T$ is the (known) first accepted one, the authors develop an inference algorithm that begins by sampling $R$.
Crucially, the inference algorithm relies on sampling the GP $f$ given $R \cup T$, which can be understood as GP binary classiﬁcation. 
We provide implementation details in Section~\ref{Survival Implementation}, including remarks on how our implementation differs from the original based on our use of network inputs.

\section{Implementation}
\label{Implementation}

In this section, we provide details for implementing the models from Section~\ref{Models}.
All of the \Rlogo \ code for our algorithms -- along with the code for reproducing our figures, simulations, and data analysis -- is available at \href{https://github.com/KolaczykResearch/GP-Networks}{github.com/KolaczykResearch/GP-Networks}.

\subsection{Binary Classification and OCC}
\label{GPC Implementation}

We develop a Gibbs sampler to overcome the intractability of our posterior.
The posterior conditional distributions are as follows:

\begin{align*}
\sigma^2 \ | \ f, \ell &\sim \text{Inv-Gamma}\Big(\alpha_\sigma + \frac{m}{2}, \ \beta_\sigma + \frac{f^TK_0^{-1}f}{2}\Big) \enskip ,\\
\hspace{.25cm} p(\ell \ | \ f, \sigma^2) &\propto |K|^{-1/2} e^{-\frac{1}{2}f^TK^{-1}f - \frac{\beta_\ell}{\ell}}\ell^{-\alpha_\ell - 1} \enskip ,\\
p(f \ | \ \sigma^2, \ell) &\propto e^{-\frac{1}{2}f^TK^{-1}f}\prod_{k=1}^m \text{logistic}(f_ky_k) \enskip .
\end{align*}

The signal variance parameter is conjugate, which makes sampling easy and in general this hyperparameter just controls the scaling of our latent function.
However, the length scale parameters are coupled with the latent function, which makes sampling more challenging.
Moreover, the determinant term adds unnecessary computational time if we elect to sample $f$ and $\ell$ separately.
Instead, we sample them jointly following \cite{murray2010slice}
Additionally, we make use of an elliptical slice sampler \citep{murray2010elliptical} for better mixing of the latent function, which adds essentially no cost as the Cholesky is already cached.
The pseudocode for our Gibbs sampler is given in Algorithm~\ref{alg: class}.

\begin{algorithm*}[!htb]
\caption{Gibbs algorithm for GP classification.}
\label{alg: class}
\SetKwInOut{Input}{Input}
\SetKwInOut{Output}{Output}
\underline{GPC Gibbs} $(K, Y, \alpha, \beta, ns)$\;
\Input{Kernel $K$, Labels $Y \in \{\pm 1\}$, inverse-gamma parameters $\alpha \ \& \ \beta$, Number of samples $ns$}
\Output{Posterior $p(f|y, \theta)$}
Initialize $\sigma^2, \ell, f$\;
\For{$t = 2,\ldots,ns$}
{
    Jointly sample $\big(f^{(t)}, \ell^{(t)}\big) \ | \ \sigma^{(t-1)}$ using slice sampler\;
    Compute $C_0$, the Cholesky of $K_0$ evaluated at $\ell^{(t)}$\;
    $C = C_0 \times \sqrt{\sigma^{(t-1)}}$, the Cholesky of $K$ evaluated at $\sigma^{(t-1)}, \ell^{(t)}$\;
    Sample $f^{(t)} \ | \ \sigma^{(t-1)}, \ell^{(t)}$ cheaply using elliptical slice sampler\;
    Sample $\sigma^2_t \sim \text{Inv-Gamma}\big(\alpha_\sigma + \frac{m}{2}, \beta_\sigma + \frac{f^TK_0^{-1}f}{2}\big)$\;
    $C = C_0 \times \sqrt{\sigma^{(t)}}$
}
\end{algorithm*}

Note that the complexity of our sampler is dominated by the Cholesky decomposition of our kernel, which is a limitation inherent to GP methods.
However, contrary to the claim in \cite{relion2019network}, our method scales well in the size of the network, because a distance matrix $D$, with $D_{ij} = d(G_i, G_j)$, only needs to be computed once and can be done in parallel.
Of course, other implementation strategies such as a grid search or maximizing the marginal likelihood would be faster than ours.
However, these approaches lose the full benefit of a Bayesian framework and in practice, we find they perform worse, as they are more limited in learning the kernel hyperparameters.

Suppose we have run Gibbs on a training set.
To make predictions, first draw $B$ samples from the posterior.
Then, we can estimate $\tilde{f}$ in two ways:

\begin{align}
\hat{\tilde{f}} &= \frac{1}{B}\sum_{b=1}^B k_{\theta^{(b)}}(\tilde{G})^TK_{\theta^{(b)}}^{-1}\hat{f} \enskip ,\\
\hat{\tilde{f}}_\text{avg} &= k_{\hat{\theta}}(\tilde{G})^TK_{\hat{\theta}}^{-1}\hat{f} \label{eq: f avg} \enskip ,
\end{align}
where $\theta^{(b)}$ is the $b^{\text{th}}$ sample of $\theta = (\sigma^2, \ell)$ and $\hat{f}$ and $\hat{\theta}$ are the estimated posterior means of $f$ and $\theta$, respectively, all of which come from the Gibbs sampler.
We recommend the plug-in estimate via the posterior mean in \eqref{eq: f avg}, which only requires one matrix inversion, as in practice we see little difference in accuracy between the versions.

\subsection{Survival Analysis}
\label{Survival Implementation}

The posterior conditional distributions are as follows:
\begin{align*}
    \sigma^2 \ | \ f, \ell_T, \ell_\text{G} &\sim \text{Inv-Gamma}\Big(\alpha_\sigma + \frac{m + \vert R \vert}{2}, \ \beta_\sigma + \frac{f^TK_0^{-1}f}{2}\Big) \enskip ,\\
    p(\ell_p \ | \ f, \sigma^2) &\propto |K|^{-1/2} e^{-\frac{1}{2}f^TK^{-1}f - \frac{\beta_p}{\ell_p}}\ell_p^{-\alpha_p - 1} \enskip ,\\
    p(f \ | \ \sigma^2, \ell) &\propto e^{-\frac{1}{2}f^TK^{-1}f}\prod_{k=1}^m \text{logistic}\big(f(T_k)\big)\prod_{r \in R_k} \big(1 - \text{logistic}(f(r)\big) \enskip .
\end{align*}
The pseudocode is given in Algorithm~\ref{alg: surv}.

\begin{algorithm*}[!htb]
    \caption{Gibbs algorithm for GP survival analysis with constant baseline hazard function.}
    \label{alg: surv}
    \SetKwInOut{Input}{Input}
    \SetKwInOut{Output}{Output}
    \underline{GP Survival} $(D, T, \alpha, \beta, ns)$\;
    \Input{Distance array $D$, Survival times $T$, Hyperparameter priors $\alpha \ \& \ \beta$, Number of samples $ns$}
    \Output{Samples from $\omega \ | \ T$}
    Initialize baseline hazard $\omega_0$ and kernel $K$\;
        \Indp $\lambda_0 = \Omega$ and $\Lambda_0(t) = \Omega \cdot t$\; \Indm
        \Indp $K_{ij} = \sigma^2\Big(\exp(-\ell_G \cdot D_{ij}^{G}) + \sum_{k=1}^p\exp(-\ell_k \cdot D_{ij}^p) + \exp(-\ell_T \cdot D_{ij}^T\Big)$\; \Indm
    Instantiate $f$ in $T$\;
    \Indp $f^{(1)}(T) \sim \mathcal{N}(0, K)$\; \Indm
    \For{$t = 2, \ldots, ns$}
      {
        \ForPar{$i = 1,\ldots, m$}
          {
            $n_i \sim \text{Poisson}(1; \Lambda_0(T_i))$\;
            $\tilde{A}_i \sim \mathcal{U}(n_i; 0, \Lambda_0(T_i))$\;
            $A_i = \Lambda_0^{-1}(\tilde{A}_i)$\;
        Sample $f(A_i) \ | \ f(R \cup T), \lambda_0$\;
        \Indp $f^{(t)}(A_i) \sim \mathcal{N}\big(k(A_i, R \cup T)^TK^{-1}f^{(q-1)}(R \cup T), \ k(A_i, A_i) - k(A_i, R \cup T)K^{-1}K(R \cup T, A_i)\big)$\; \Indm
            $U_i \sim \mathcal{U}(n_i; \ 0, 1)$\;
            $R_i = \{a \in A_i \text{ such that } U_i < 1 - \sigma(f(a))\}$\;
          }
        $R = \bigcup_{i=1}^m R_i$\;
        Update parameters of $\lambda_0$\;
        \Indp $\Omega \sim \Gamma\big(\alpha_\Omega + m + \vert R \vert, \beta_\Omega + \sum_{i=1}^m T_i\big)$\; \Indm
        Update $f(R \cup T)$ and hyperparameters of kernel as in Algorithm~\ref{alg: class}\;
            \Indp Jointly sample $f(R \cup T)$ and length scales for time and covariates using slice sampler\; \Indm
            \Indp Sample $f(R \cup T)$ several times using elliptical slice sampler\; \Indm
            \Indp Sample signal variance from conjugate posterior\; \Indm
      }
\end{algorithm*}

Note that $f$ changes size at each iteration of the sampler: at $t = 1$, $R = \emptyset$ and $\vert f \vert = m$, whereas for $t \geq 2$, we have $\vert f \vert = m + \vert R \vert$.
To update $f(R \cup T)$ in our slice sampler, we input the concatenation of $f(T)$ with $f(R)$, where the latter is from our previous sample $f(A)$.
Also, notice that the sampling of rejected points and their imputation given $f(R \cup T)$ can be implemented in parallel, which was not discussed in the original paper \citep{fernandez2016gaussian}.
In addition, it is important to remember that each time point has associated covariates that are omitted in the notation, i.e. $f(t) = f(t, X)$, where $X$ in this case is a network.
Therefore, it is necessary to keep track of the implicit covariates when sampling $A$.
Finally, as the authors point out, setting $R$ can be seen as a GP classification problem.
This procedures essentially uses a noisy version of our classifier from Section~\ref{Binary Classification}, which would be too computationally expensive to use at every step of the Markov chain.

One computational note is that Algorithm~\ref{alg: surv} can be slow if $\vert R \vert$ is large, which the authors circumvent using a kernel approximation.
Unfortunately, such approximations are not available for network kernels because they rely on Bochner's theorem and Fourier transforms of kernels on $\mathbb{R}^d$.
Another computational note is that the kernels, while theoretically positive definite, can be numerically close to singular.
For this reason, it is common to add jitter to the diagonal of the kernel to improve the condition number.
We find that more jitter is needed in the survival analysis setting than the binary classification, which is likely due to the fact that the rejected points $R$ share all of the covariates with one of the original points, thereby introducing dependency into the rows and columns of the Gram matrix.

Having run the sampler, the survival surfaces are given as
\[S(t, X) = \exp\Big(-\int_0^t \omega(s, X) \ ds\Big) \enskip .\]
We use the trapezoidal rule for integration by evaluating $f(s, X)$ on a uniform grid from $s = 0$ to $s = \max T$ with $\Delta = \frac{\max T}{99}$, i.e.
\[\int_0^t \omega(s, X) \ ds \approx \frac{\Delta}{2} \sum_{k=1}^{100} \omega(s_{k-1}, X) + \omega(s_k, X) \enskip .\]

\section{Theoretical Results}
\label{Theory}

In this section, we prove that our graph distances in \eqref{dist:Frobenius} and \eqref{dist:lambda} induce valid kernels.
Using these kernels, we then prove that our binary classifier from Section~\ref{Binary Classification} and our survival analysis model from Section~\ref{Survival Analysis} both achieve posterior consistency.

\subsection{PD Kernels}
\label{PD Kernels}

Central to the GP framework is a kernel for capturing the dissimilarity between data points.
However, not all dissimilarity measures will induce a valid kernel, i.e a provably positive definite Gram matrix.
This problem is further complicated by our use of network objects as it is not obvious which graph distances will ensure a valid kernel.

In order to work with the squared-exponential kernel, we need a distance that is an inner product on the space of graphs.
One familiar choice is the Hamming distance or Frobenius distance, which is typically defined for two binary networks as the sum of the absolute differences of their corresponding adjacency matrices.
However, as we will see in our proof of consistency, we need to work in the space of weighted networks.
Moreover, we often encounter weighted networks, so using the edge weights rather than indicators for edges may provide us more information and we avoid the common challenge of determining a threshold value for edges.
In this setting, we can view binary networks as having weights $\{0, 1\} \subset [0, 1]$.

In order to work with weighted networks, we use the Frobenius distance in \eqref{dist:Frobenius}, which we recall as
\begin{equation*}
    d_{\text{F}}\big(G_1, G_2\big) := \frac{1}{n(n-1)}\sum_{i, j \in V} \big(A_{ij}^{(1)} - A_{ij}^{(2)}\big)^2 \enskip .
\end{equation*}
Note that for binary networks, the sum of absolute differences is equal to the sum of squared differences, so our Frobenius distance is consistent with the commonly used Hamming distance for binary networks.
Furthermore, the Frobenius distance can be used for graphs with negative weights such as measures of correlation.

Another choice of graph distance is in \eqref{dist:lambda}, which we recall as
\begin{equation*}
    d_\lambda\big(G_1, G_2\big) := \sum_{i=0}^{n-1}\big\vert \lambda_i^{(1)} - \lambda_i^{(2)}\big\vert^2\enskip .
\end{equation*}
This is a spectral-based distance, which has the benefit of capturing more of the network structure at the cost of more computational complexity.

\begin{theorem} \label{thm: PD}
If $G_1, \ldots, G_m$ are weighted networks and
\[K_{ij} = \sigma^2 \exp\Big(-\ell \cdot d\big(G_i, G_j\big)\Big) \enskip ,\]
with $d(\cdot, \cdot)$ given in \eqref{dist:Frobenius} or \eqref{dist:lambda}, then $K$ is a positive-definite kernel.
\end{theorem}

\begin{proof}
\cite{jayasumana2013kernel} provide the following theorem.
Let $(M, d)$ be a metric space and define $k : X \times X \to \mathbb{R}$ as $k(x_i, x_j) = \exp\big\{-d^2(x_i,x_j)/2\sigma^2\big\}$.
Then $k$ is positive definite for all $\sigma > 0$ if and only if there exists an inner product space $V$ and a function $\psi : M \to V$ such that $d(x_i,x_j) = \Vert\psi(x_i) - \psi(x_j)\Vert$.

Using the Frobenius distance and absorbing the normalizing constant into $\sigma^2$, we have
\begin{align*}
d_\text{F}\big(G_1, G_2\big) &= \sum_{i, j \in V} \big(A_{ij}^{(1)} - A_{ij}^{(2)}\big)^2 \\
&= (A^{(1)}-A^{(2)}) \cdot (A^{(1)}-A^{(2)}) \\
&= \Vert A^{(1)}-A^{(2)} \Vert_\text{F}^2 \enskip ,
\end{align*}
where $\cdot$ and $\Vert \cdot \Vert_\text{F}$ are the Frobenius inner product and norm, respectively, and $\psi(\cdot)$ is the identity map.

Similarly, using the spectral-based distance, we have
\[
d_\lambda\big(G_1, G_2\big) = \sum_{i=0}^{n-1}\big\vert \lambda_i^{(1)} - \lambda_i^{(2)}\big\vert^2 = \Vert \lambda^{(1)} - \lambda^{(2)} \Vert^2 \enskip ,
\]
where $\psi(\cdot)$ is the map taking the graph's adjacency matrix to the eigenvalues of its Laplacian.
\end{proof}

\subsection{Posterior Consistency of
Classifier}

Let $G_1, \ldots, G_m$ be simple weighted networks on the same set of $n$ vertices, $V$.
That is, $G_k = (V, E^{(k)}, W^{(k)})$ for all $k = 1, \ldots, m$ with $w_{jj} = 0$ for all $j = 1, \ldots, n$.
Unlike many graph results that consider the asymptotic behavior as the order of the network, $n$, grows, we return to the classical data setting and investigate what happens as the sample size, $m$, of observed networks increases.
We expect that this asymptotic regime will become increasingly important in the network literature as it is likely, for example in the case of connectomics, that we sample more brain networks with a fixed number of regions of interest.

With the set up from Section~\ref{Binary Classification} with inverse-gamma hyperpriors over the kernel induced by either distance \eqref{dist:Frobenius} or \eqref{dist:lambda}, we show that our classifier achieves posterior consistency.

\begin{theorem} \label{thm: class cons}
GP classification with network inputs achieves posterior consistency for the squared-exponential kernel.
\end{theorem}

We defer the proof to the \hyperref[pf: class cons]{Appendix}.
Ultimately, we show that our classifier achieves posterior consistency by verifying the conditions given in \cite{ghosal2006posterior}.
Here, we make a few remarks about the proof.

First, the infinite differentiability of our kernel is important for verifying several of the conditions, which demonstrates the importance of using a squared-exponential kernel rather than, say, a Mat\'{e}rn kernel.
This reiterates the need to have a graph distance that guarantees positive-definiteness of our squared-exponential kernel.
Secondly, it is important that we bound our covariate space, which forces us to work in the space of weighted networks.
Finally, the inverse-gamma hyperpriors ensure full support over the kernel hyperparameters.
Although other fully supported priors could be used while maintaining posterior consistency, inverse-gamma is conjugate for the signal variance hyperparameter.

\subsection{Posterior Consistency of Survival Analysis}

As above, let $G_1, \ldots, G_m$ be simple weighted networks on the same set of vertices.
Consider the survival analysis model from Section~\ref{Survival Analysis} with the kernel in network and time given by \eqref{eq: K surv}.
Under the assumption that the survival times have finite expectation, we show that our survival analysis model achieves posterior consistency in the limit of $m$.

\begin{theorem} \label{thm: surv cons}
Under mild assumptions, survival analysis with network inputs achieves posterior consistency.
\end{theorem}

Again, we defer the proof to the \hyperref[pf: surv cons]{Appendix}, which consists of verifying the conditions given in \cite{fernandez2016posterior}.
The main assumption that we make is that the mean of the survival time is finite, which is true for events like death or gestational age and reasonable for many others.
Otherwise, similar to Theorem~\ref{thm: surv cons}, we rely on our use of a squared-exponential kernel due to its stationarity.

\section{Simulations}
\label{Simulations}

In this section, we evaluate our models for binary classification, anomaly detection, and survival analysis in several simulated settings.

\subsection{Binary Classification}
\label{simulation: class}

We first present the results of a simulation study to assess the accuracy of our GP classifier.
Since classification is at the center of our anomaly detection and survival analysis models, it is important to identify what scenarios are better suited to one or another kernel or distance.

We consider the squared-exponential kernel with both the Frobenius distance in \eqref{dist:Frobenius} and the spectral-based distance in \eqref{dist:lambda} as well as a $k$-step random walk kernel in \cite{vishwanathan2010graph}.
We refer to these methods as GP-F, GP-$\lambda$, and GP-RW, respectively.
To isolate the contribution of our method, we also include a comparison to an SVM with the same random walk kernel.
We use a cross-validation procedure to choose the cost penalty of the SVM from the grid suggested by \cite{hsu2003practical} and we refer to this method as SVM-RW.
Finally, we also include two recent graph classifiers from \cite{relion2019network} and \cite{arroyo2019inference}, which we refer to as GC and MASE.

Our simulation is a fully crossed design varying the sample size ($m = 40, 100, \ldots, 340$) and the network size ($n = 20, 40, \ldots, 100$), as well as the random graph model from which we sample the network inputs.
We compare small-world networks \citep{watts1998collective} with different rewiring probabilities $p0 = .05$ and $p1 = .07$, networks from a stochastic block models with two equal-sized communities and probability link matrices given as
\begin{equation*}
    G \ | \ Y = -1 \sim \text{SBM}\begin{pmatrix}
    .05 & .15 \\
    .15 & .05
    \end{pmatrix} \quad \text{ and } \quad
    G \ | \ Y = +1 \sim \text{SBM}\begin{pmatrix}
    .1 & .15 \\
    .15 & .05
    \end{pmatrix} \enskip ,
\end{equation*}
and correlated Erd\H{o}s-R\'{e}nyi networks \citep{pedarsani2011privacy}
in which we sample independent Erd\H{o}s-R\'{e}nyi networks, $G_0, G_1 \sim ER(n, p = .8)$ and then generate $G \ | \ Y = -1$ as
    \[G_{ij} \ | \ Y = -1 \sim \begin{cases}
        \text{Bernoulli}(\rho = .8) &\text{ if } G_{0, i j} = 1\\
        0 &\text{ if } G_{0, i j} = 0
    \end{cases} \enskip ,\]
and similarly for $G \ | \ Y = +1$ conditional on $G_1$.
We also compare networks from preferential attachment models \citep{barabasi1999emergence} with different powers $\alpha_0 = .6$ and $\alpha_1 = 1.4$ and exponential random graph models that have the same average density and 2-stars but different transitivity.
We refer to these five random graph models as small-world, SBM, corrER, PA, and ERGM, respectively.

For each network model, we sample $m/2$ networks from each class and use a 75/25 train/test split to assess the classification accuracy.
Results are given as the average accuracy on the test set over 50 replicates in Figure~\ref{fig: class sim}, but the overall conclusions are as follows.
For global structure like the small-world, PA, and ERGM networks, GP-$\lambda$ is the better graph distance, whereas for local structure like the correlated Erd\H{o}s-R\'{e}nyi networks, GP-F is better.
This is unsurprising since the spectral-based distance and Frobenius distance were expected to capture global and local differences, respectively.
Note that for GP-$\lambda$, the normalized Laplacian consistently outperformed the Laplacian, so we only report those results and recommend always using the normalized Laplacian in practice.
We also see that the GP-RW is as good or better than the SVM-RW in all of the graph models.
Finally, MASE is an excellent method in the case when the underlying network structure is low-rank like the SBM networks.
However, we note that we encountered a few instances of eigendecomposition failure with MASE and those replicates are not included in the results.

\begin{figure}[!htb]
\centering
\includegraphics[width=\textwidth]{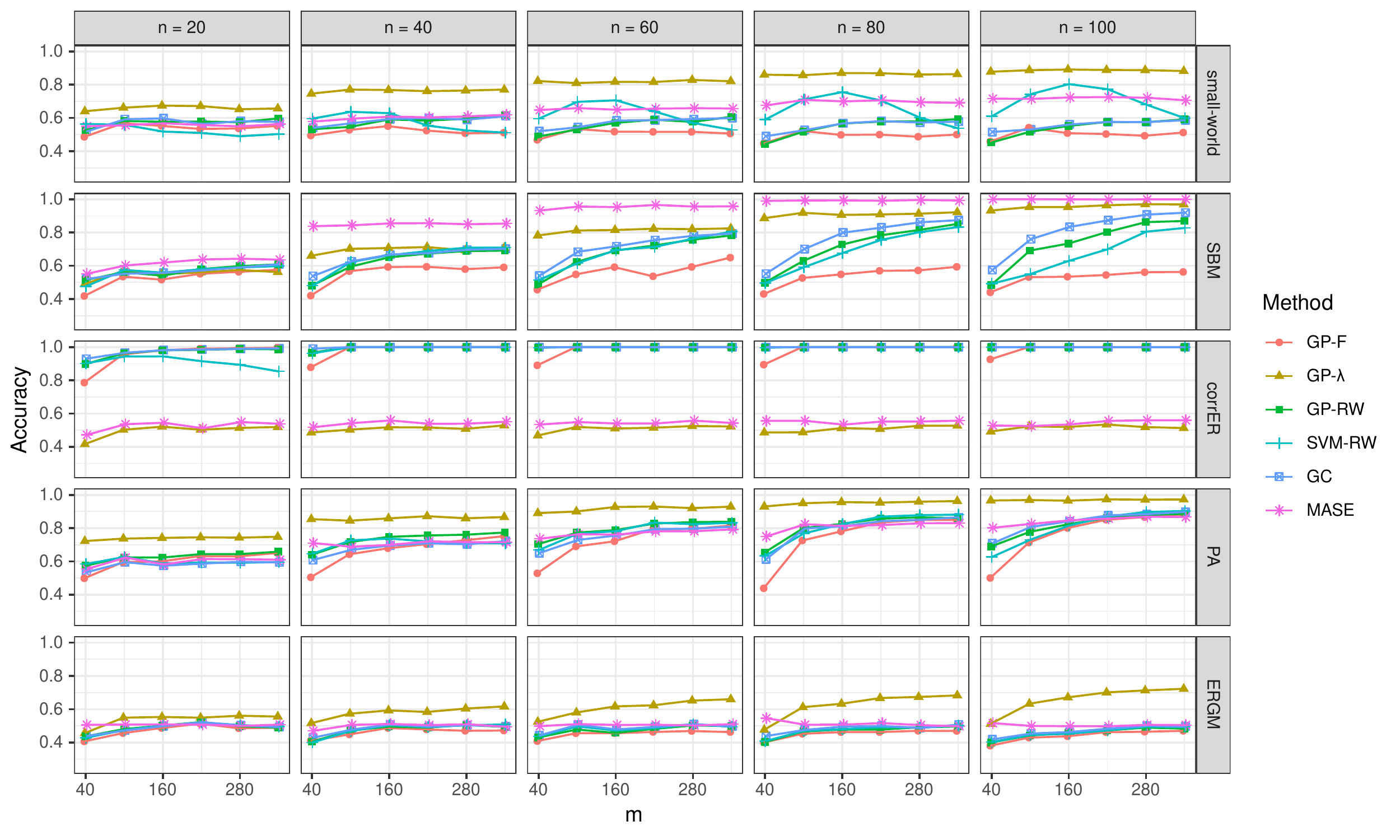}
\caption{Simulation results for three different random graph models. On the y-axis, we have classification accuracy using a held out test set. On the x-axis, we vary the number of networks from $m = 40$ to $m = 340$, which are split into two even groups. The column panels correspond to networks with $n$ nodes.}
\label{fig: class sim}
\end{figure}

In the first row of Figure~\ref{fig: class sim}, we see that GP-$\lambda$ performs the best for small-world networks across all values of $m$ and $n$.
As $n$ increases, most methods improve with MASE the second most accurate.
GP-RW and GC perform similarly.
Note that GP-RW starts to outperform SVM-RW when the ratio $m/n$ exceeds 3-5.
This may be due to the fact that we did not scale the grid as the number of nodes increased, but nevertheless this highlights the flexibility of our GP model as we use the same hyperprior for all of the settings.
Finally, GP-F is consistently the worst classifier and shows little improvement for different $m$ and $n$, which is unsurprising since small-world networks are unlikely to be characterized by local differences.

In the second row, we see that MASE performs the best, though GP-$\lambda$ closes the performance gap as $n$ increases.
Similarly, GC is the third best classifier with performance increasing in both $m$ and $n$.
We also see a clear ranking of GP-RW, then SVM-RW, then GP-F, which again highlights the superiority of a GP classifier over an SVM despite the similarity of their objective functions.

The third row shows the results for corrER networks.
Note that by design, these networks are correlated, so this serves as a check of the robustness to potential violations of independence assumptions.
GC and GP-RW rapidly approach perfect classification accuracy for the different values of $m$ and $n$.
Similar for SVM-RW except for $n = 20$ and GP-F except for small $m$.
Interestingly, both MASE and GP-$\lambda$ perform very poorly on this task and no not improve as $m$ and $n$ increase.

In the fourth and fifth rows, there are no clear trends other than GP-$\lambda$ outperforming all of the other methods.
For the PA networks, all of the methods improve more with $n$ than with $m$.
However, none of the methods besides GP-$\lambda$ is successful at classifying the ERGM networks, which appears to be a challenging task.

We also ran our classifiers on the COBRE dataset from \cite{relion2019network}.
Due to their negative weights, only GP-F is suitable to use these networks directly.
To use GP-$\lambda$, we can represent each network by its signed Laplacian \citep{kunegis2010spectral},
    \begin{equation}\label{eq: signed laplacian}
        \bar{L}_{ij} = \begin{cases}
        -A_{ij} & i \neq j \\
        \sum_{k=1}^n |A_{ik}| - A_{ii} & i = j
        \end{cases} \enskip ,
    \end{equation}
and use $\bar{L}$ in \eqref{dist:lambda}.
Finally, for GP-RW, we need to binarize the network and we use $0.45$ as a threshold for an edge resulting in $5.3\%$ of the original edges, which closely matches the best case ($5.4\%$) reported in \cite{relion2019network}.

Using the same 10-fold cross-validation as reported by the authors, we obtained a classification accuracy of $.92$ with a standard deviation of $.10$ for GP-F and also a mean AUC of $.98$ with a standard deviation of $.05$.
For GP-$\lambda$, we obtain an accuracy of $.56$ $(.17)$ and an AUC of $.63$ $(.21)$.
Finally, for GP-RW, using the binarized networks, we obtain an accuracy of $.75$  $(.06)$ and an AUC of $.90$ $(.09)$.
These results suggest that some information may be lost by thresholding the edge weights and also that the differences between the schizophrenic and the control networks may be more local than global.

\subsection{Anomaly Detection}

Next, we  emulate the scenario when one class is much more prevalent yet we would like to be able to identify possible samples from another class perhaps as new samples become available.
To do this, we set the classes as unbalanced (90/10) and train our classifier exclusively on one class, but still use a 75/25 train/test split for evaluating the performance.

From Figure~\ref{fig: class sim}, we see that the SBM can be challenging to classify for small $m$ even when the classes are balanced.
We use this regime with $m = 100$ and $n = 100$ to illustrate our ability to perform anomaly detection.
Note that even though GC provides classification probabilities, it cannot be trained on a single class.
Likewise for SVM-RW and MASE.
Therefore, we only report the results for GP-F, GP-$\lambda$, and GP-RW, which are given in Figure~\ref{fig: occ sim}.

\begin{figure}[!htb]
\centering
\includegraphics[width=\textwidth]{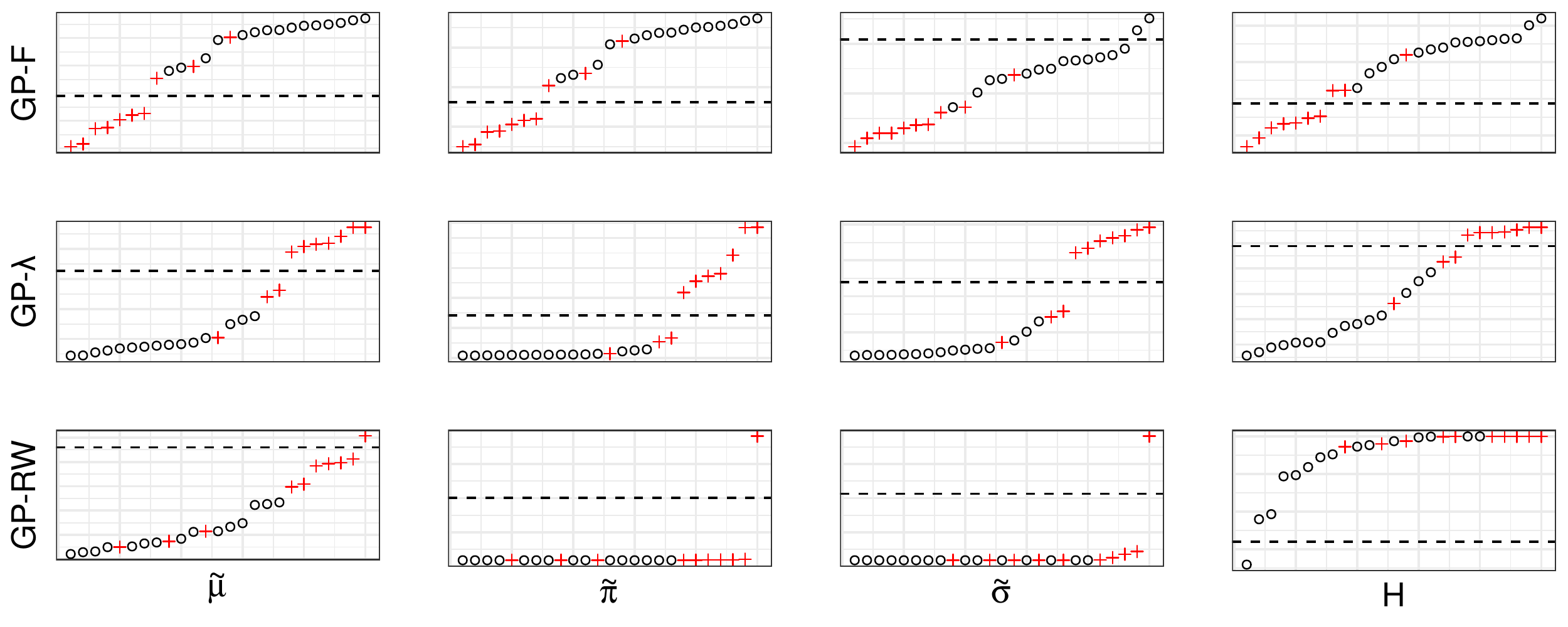}
\caption{Four different OCC scores defined in Section~\ref{Anomaly Detection} applied to our simulation test set. Scores are derived from the posterior of our GP classifier trained only on one class (black circle). Dashed lines indicate possible discrimination thresholds based on the elbow method, i.e. the largest jump among the sorted values.}
\label{fig: occ sim}
\end{figure}

We see from the middle row that GP-$\lambda$ has the most discriminatory power for all of the scores with good separation using the elbow method.
GP-F also has good discrimination power, though the scores are more clumped and the ordering is reversed.
Although GP-RW sorts the scores well, there is little separation using the elbow method for any of the scores.
Overall, the class probabilities $\tilde{\mu}$ and the coefficient of variation $H$ are the best scores to use.

\subsection{Survival Analysis}

Finally, we consider two survival analysis problems.
In the first case, we sample $m$ points each from $f_0(t) = \mathcal{N}(2, 0.8^2)$ and $f_1(t) = \mathcal{N}(4, 1)$, restricted to $\mathbb{R}^+$.
However, instead of providing an indicator for which density a point was sampled from, we use as input a network sampled from $ER(p_0)$ or $ER(p_1)$.
In this way, we simultaneously estimate the survival curves and perform classification.
Moreover, we reduce to the noiseless case in the limit as $p_0 \to 0$ and $p_1 \to 1$.
In the second case, we consider the more challenging variant from \cite{fernandez2016gaussian}, in which $f_0(t) = \mathcal{N}(3, 0.8^2)$ and $f_1(t) = 0.4\mathcal{N}(4, 1) + 0.6\mathcal{N}(2, 0.8^2)$.
This is similar to the first case except that $f_1(t)$ is now a mixture of the Gaussians in the first case, which creates a crossing between the survival functions corresponding to $f_0(t)$ and $f_1(t)$.

In both cases, we take $m = 50$ samples each from $f_0(t)$ and $f_1(t)$, where the network inputs are sampled from an $ER$ with $n = 50$ nodes and $p_0 = .3$ and $p_1 = .7$.
These classes are very easy to distinguish for both GP-F and GP-$\lambda$, so we only show the results for GP-F.
Also, recall that our kernel must be stationary for the survival function to be well-defined, so we cannot use GP-RW.

The results are shown in Figure~\ref{fig: surv sim}.
We see that we successfully estimate $f_0$ and $f_1$ in the easy case.
On the other hand, the estimates in the hard case are essentially the mean of $f_0$ and $f_1$, which are the same results as \cite{de2009bayesian}, who originally introduced this experiment.
Unfortunately, we would need a much larger $m$ to recover the crossing survival functions, which is not feasible as our algorithm scales poorly in $m$ even though parallelization helps.
Although limiting, we note that our application only has a total of $m = 37$.

\begin{figure}[!htb]
\centering
\includegraphics[width=\textwidth]{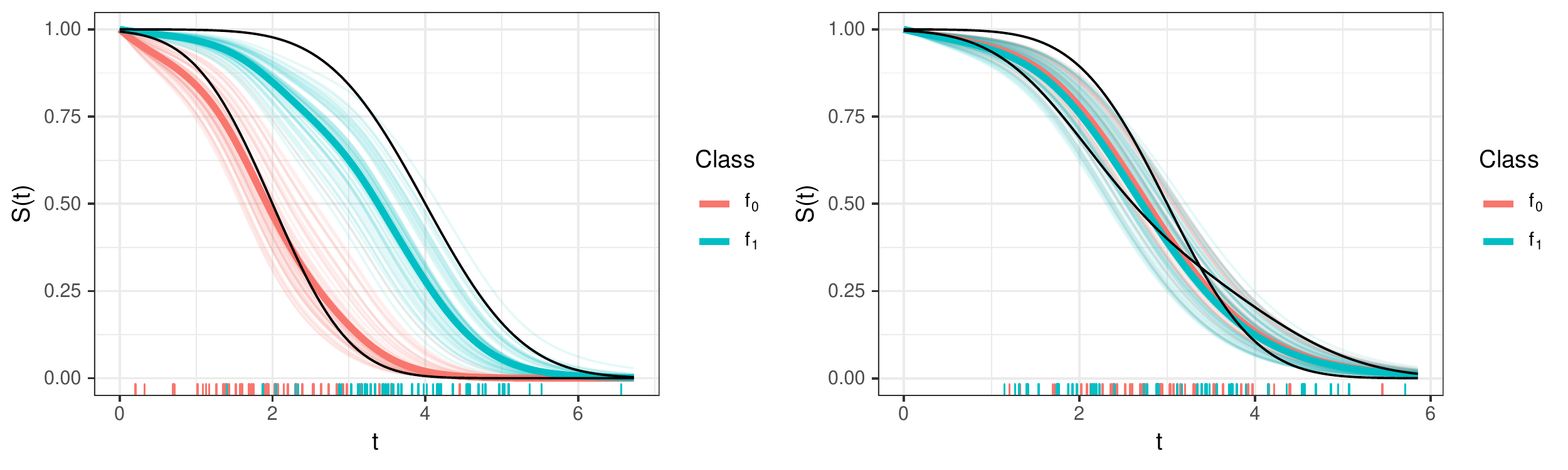}
\caption{Simulation results for the easy case (left) and hard case (right). The light curves are the posterior means of survival surfaces $S(t, G)$ for each point, the darker curves are survival curves $S(t)$ of each group and the black curves are the true survival curves. We sample $m = 50$ points from each density and corresponding ER networks with $n = 50$ nodes with $p_0 = .3$ and $p_1 = .7$.}
\label{fig: surv sim}
\end{figure}

\section{Application to Microbiome Networks}
\label{MB}

Finally, we return to our motivating application of the microbiome dataset from \cite{digiulio2015temporal}.

\subsection{Microbiome Network Construction}\label{MB_nets}

For every subject in the dataset, there is an operational taxonomic unit (OTU) table of genomic count data.
The rows of an OTU table correspond to the different samples over the course of a pregnancy and the columns represent different bacterial taxa.
The $(i, j)$-th entry is therefore the count of how much genomic material of taxa $j$ was found in the sample at time $i$.

Given an OTU table, there are many methods for constructing microbiome networks.
Typically, the nodes of a microbiome network represent a taxa and the edge weights between nodes represent some association between taxa with different construction methods reflecting different choices of association.
Herein, the edges of our microbiome networks are correlations, specifically Sparse Correlations for Compositional data (SparCC) following the well-known procedure in \cite{friedman2012inferring}.
SparCC approximates the Pearson correlations between the log-transformed counts with the assumption that the true correlation network is sparse, though the authors argue their method is robust to misspecification of the sparsity assumption.

In order to employ the SparCC method, we need to make several preprocessing decisions on what samples to include.
To do so, we follow similar guidelines as \cite{bogart2019mitre}, who use the same dataset, among others, to demonstrate their Bayesian method called MITRE for microbiome time-series data.
In particular, we discard potentially spurious taxa with fewer than 500 total samples and exclude samples where coverage is lower than 10 total samples.
We include only samples from the vaginal swabs and restrict the samples to between 1 and 30 gestational weeks.
While the networks may have more information if we include samples from all body sites and later time points, we want to illustrate how they would be used ideally in practice, which is as a minimally-invasive monitoring procedure throughout pregnancy that would enable intervention in the case that a preterm delivery may be likely.

Following this construction, we obtain $m = 37$ networks with $n = 63$ nodes each.
Two such networks are shown in Figure~\ref{fig: networks}.

\subsection{Exploratory Analysis}\label{EDA}

As with studies of non-network data, it is important to first perform an exploratory analysis.
With network data, it is common to try to reduce networks to a few summary statistics.
These summary statistics have been developed to describe aspects such as the overall network structure, local connectivity patterns, and important nodes.
Although many such statistics have analogues for weighted networks, the edge weights in our microbiome networks are correlations and hence can be negative.
Unfortunately, methods for networks with negative weights is an underdeveloped area requiring us to either use absolute values, which fails to distinguish between negative and positive associations, or else binarize the weights.
Since thresholding is difficult without having prior information, we use the weight magnitudes for our exploratory analysis.

In general, the challenge with describing a network is identifying what statistics are relevant for a given task.
In the networks in Figure~\ref{fig: networks}, there appear to be some visual differences between the preterm and term networks such as the apparent densities and degree distributions.
To assess this, we compare the averages in the preterm and term groups across several summary statistics.

We consider four network-level descriptors.
We first measure strength, which is the sum of all of the edge weights in the network.
Second, we estimate the number of communities in each network, which we denote $\hat{K}$.
Third, for distance measures, we use the inverse of the edge weights so that higher correlations refer to closer nodes.
This allows us to measure closeness, a node-level statistic for how many steps it takes to reach every other node from a given node, which we aggregate over all nodes.
Our final network-level statistic is diameter, which is the length of the longest path.

We also consider four node-level descriptors.
We look at what percentage of the networks have \textit{Gardnerella} or \textit{Ureaplas}, which are both reported in \cite{digiulio2015temporal} to be associated with higher risk of preterm delivery, in their top three most central nodes.
Similarly, we report the percentage for any \textit{Lactobacillus} species, which were reported to be inversely correlated with gestational length, in the top three central nodes.
We measure centrality by eigenvalue, but the results are similar for centrality by strength, closeness, and hub score.
We also report the percentages for \textit{Finegoldia}, which we have identified as a central species for term networks that is absent in preterm networks.
This species was not identified previously and deserves further scientific inquiry.
The results are given in Table~\ref{tab:EDA}.

\begin{table}[!htb]
\begin{tabular}{@{}llll@{}}
\toprule
\rowcolor[HTML]{FFFFFF} 
\multicolumn{1}{c}{\cellcolor[HTML]{FFFFFF}\textbf{Characteristic}} &
  \multicolumn{1}{c}{\cellcolor[HTML]{FFFFFF}} &
  \multicolumn{1}{c}{\cellcolor[HTML]{FFFFFF}\textbf{Term (m = 26)}} &
  \multicolumn{1}{c}{\cellcolor[HTML]{FFFFFF}\textbf{Preterm (m = 11)}} \\ \midrule
\rowcolor[HTML]{EFEFEF} 
Network-level &                        &              &             \\
              & Strength               & 709 (784)    & 490 (231)   \\
\rowcolor[HTML]{EFEFEF} 
              & $\hat{K}$              & 3.19 (0.69)  & 3.45 (0.82) \\
              & Closeness              & 0.21 (0.20)  & 0.15 (0.06) \\
\rowcolor[HTML]{EFEFEF} 
              & Diameter               & 14.0 (6.1)   & 15.0 (5.3)  \\
Node-level    &                        &              &             \\
\rowcolor[HTML]{EFEFEF} 
              & \textit{Gardnerella}   & 1/26 (4\%)   & 3/11 (27\%) \\
              & \textit{Ureaplas}      & 2/26 (8\%)  & 2/11 (18\%) \\
\rowcolor[HTML]{EFEFEF} 
              & \textit{Lactobacillus} & 19/26 (73\%) & 2/11 (18\%) \\
              & \textit{Finegoldia}    & 5/26 (19\%)  & 0/11 (0\%)  \\ \bottomrule
\end{tabular}
\caption{
\begin{flushleft}
Comparison of average summary statistics for preterm and term networks. The average (sd) is reported for the network-level statistics and the proportion of networks with the given species in the top three most central taxa is reported for the node-level statistics.
\end{flushleft}
}
\label{tab:EDA}
\end{table}

Although we see some network-level differences such as strength, none are significant.
This underscores the topological complexity of the microbiome networks that are clear from Figure \ref{fig: networks}.
We also have some validation of our network construction by seeing that the keystone species identified in the literature manifest as differences in node centrality.
Moreover, we are also able to identify a difference in an unreported species.
However, without \textit{a priori} knowledge of the keystone species, there is an inherent multiple testing challenge in identifying potentially relevant taxa in new microbiome networks.
Therefore, the entire networks should ideally be used directly as inputs.

\subsection{Preterm Delivery Tasks}

To apply our models from Section~\ref{Models}, we define three tasks based on the additional clinical information available in the dataset.
In particular, there is an indicator for a history of preterm delivery, an indicator if the given pregnancy resulted in a preterm birth, and the length of pregnancy.
This allows us to first define a classification problem using preterm delivery status.
Note that \cite{bogart2019mitre} perform the same classification task, but only group the preterm and very preterm labels in the dataset as preterm, whereas the original authors included the marginal label in the preterm class, which they define as before the 37th gestational week.

Next, we note that 11 of the $m = 37$ pregnancies in the case-control cohort resulted in a preterm delivery, but this proportion of preterm deliveries is higher than the national average of 11\%.
For this reason, we train our classifier using a more realistic term/preterm balance and then perform anomaly detection.
We expect this is more likely how our models would be used in practice i.e given a microbiome network during pregnancy, scores are output and monitored rather than probabilities.
Finally, we perform survival analysis using days to delivery as our event.

\subsection{Results}

Using leave-one-out cross validation (LOOCV), GP-$\lambda$ using the signed Laplacian in \eqref{eq: signed laplacian} has a classification accuracy of 78\% with an AUC of $0.69$, whereas GP-F has an accuracy of 61\% and an AUC of $0.37$.
GP-RW using binarized networks has an accuracy of 70\% with an AUC of $0.53$, but unlike GP-$\lambda$ and GP-F, it predicts term for every subject.
This underscores the challenge of this classification task.
In particular, LOOCV exacerbates the problem of unbalanced classes, since every time a subject with a preterm label is withheld, the classes become less balanced.
Consequently, metrics like AUC may not be useful in assessing the classifiers.

Instead, we reformulate the problem as an anomaly detection task by setting our training set to reflect the population, i.e. roughly 11\% of training cases are preterm.
We then run our GP classifier on this unbalanced training set to obtain OCC scores.
The results of this are shown in Figure~\ref{fig: mb occ}.

\begin{figure}[!htb]
\centering
\includegraphics[width=\textwidth]{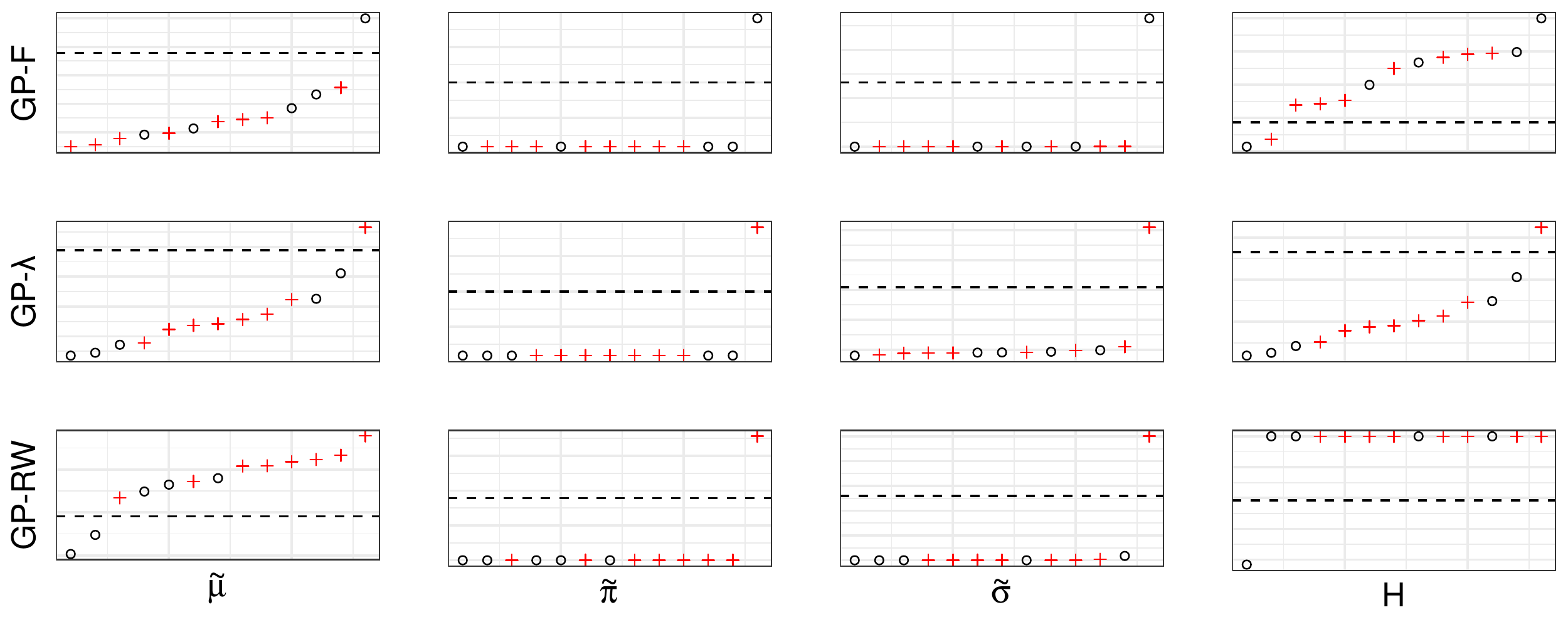}
\caption{OCC scores for the microbiome data. Preterm births are denoted by a red plus sign and the dashed lines indicate the largest gap in values.}
\label{fig: mb occ}
\end{figure}

Based on the results, $\tilde{\mu}$ from GP-RW would be most useful.
Next, the low-end scores of $\tilde{\mu}$ and $H$ from GP-F as well as the high-end scores of $\tilde{\sigma}$ from GP-$\lambda$
If microbiome sampling were to become part of prenatal care, then these scores could be used to flag high risk pregnancies rather than using the probabilities of preterm delivery.
The results also suggest that differences in the microbiome networks for preterm and term pregnancies are a mix of local and global differences, but perhaps more global due to the slightly higher performance of GP-$\lambda$ over GP-F.
This is consistent with the belief that microbiome networks are complex while simultaneously supporting the original study findings regarding pronounced differences in a few taxa.

Finally, we perform survival analysis on time to birth, where we use as input both the subject's microbiome network and an indicator for history of preterm birth.
Here, we are able to seamlessly combine data of mixed types using our graph kernel with the Frobenius distance and a squared-exponential kernel with Euclidean distance.
The result with the spectral-based distance is shown in Figure~\ref{fig: mb surv}, which includes a comparison with the Kaplan-Meier (KM) estimates that only include the indicator as input.

\begin{figure}[!htb]
\centering
\includegraphics[width=\textwidth]{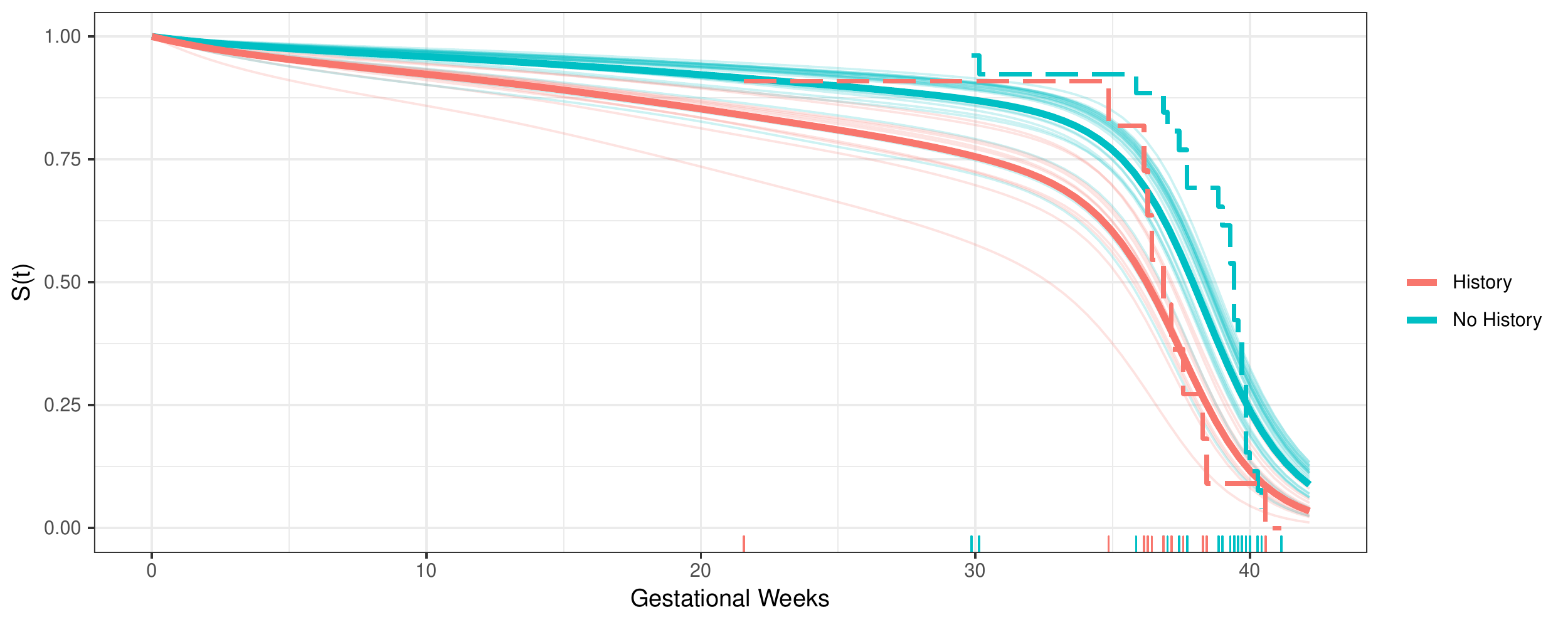}
\caption{The light curves are the posterior means of survival surfaces $S(t, G)$ for each subject, the darker curves are survival curves $S(t)$ for those with and without history of preterm birth, and the dashed curves are Kaplan-Meier estimates.}
\label{fig: mb surv}
\end{figure}

One obvious difference between our GP method and the KM estimates is that we are able to incorporate relevant covariates like the microbiome network and history of preterm birth.
This ability is necessary in practice because usually there are many other relevant clinical characteristics that need to be accounted for when personalizing care.
Another difference between our GP method and the KM estimates is that smoothing allows us to extrapolate the likelihood of delivery prior to any observed time points.
This is crucially important to clinicians because longer gestational periods are directly related to higher survival rates in pregnancies at higher risk of preterm delivery.

\section{Conclusion}
\label{Conclusion}

\subsection{Further Applications}

In Section~\ref{MB}, we saw how our unified Bayesian framework for classification, anomaly detection, and survival analysis could be applied to a unique microbiome dataset.
Originally, \cite{digiulio2015temporal} demonstrated that microbiome networks captured the intricate dynamics between microbial taxa in pregnant women throughout gestation.
Our work now represents an important advance in this research by providing statistical methodology for incorporating this data as network covariates.
Consequently, we hope that this paper convinces practitioners across various scientific disciplines to collect network-based data objects and allows them to pose scientific questions using these networks as covariates.
To the best of our knowledge, there are currently quite few datasets with individual networks and outcomes for all three problems of classification, anomaly detection, and survival analysis.

We believe that these tools will empower ecologists to design microbiome studies with many subjects across many time points, enabling further studies of relationships between the human microbiome and conditions such as pregnancy or various disease states.
Additionally, these tools could be applied to a variety of non-human microbiomes and soil microbiomes, which are active areas of research.

Beyond ecology, our work is applicable to other fields that may benefit from the use of a network per individual or unit.
As technological advances are making this increasingly possible, several areas in particular are poised to utilize these methods.
For example, we have already shown that our models are useful in neuroscience for classifying diseases using fMRI networks.
Furthermore, these methods can help address clinically relevant time-to-event questions about the brain such as the development of Alzheimer's disease.
Other areas of application include the detection of cancer using gene co-expression networks or single-cell networks, as well as social network analysis for identifying early adopters or performing time-to-purchase analyses using ego-nets.

\subsection{Future Work}
There are many exciting extensions to this work.
First, it would be impactful if we had graph distances for networks of different orders or directed networks.
Next, we have seen that kernels seamlessly combine data of mixed types, such as patient networks and clinical information, but it would be valuable to have kernels that incorporate exogenous network information.
Similarly, approximations of graph kernels in the spirit of \cite{rahimi2008random} would speed up computations and may be necessary as networks become more present and sample sizes increase.
For survival analysis, it may be possible to adapt the recent variational approach by \cite{kim2018variational} to network inputs.
And, as always, we would benefit from better and faster MCMC or other Bayesian computational methods.


\bibliographystyle{imsart-nameyear}
\bibliography{references.bib}

\newpage
\begin{appendix}
\label{Appendix}

\section*{}

\begin{proof}[Proof of Theorem~\ref{thm: class cons}] \label{pf: class cons}
Let $\varepsilon > 0$.
We want to show that
\[\Pi\Big( p : \int |p(x) - p_0(x)| dQ(x) > \varepsilon \Big\vert Y_1,\ldots,Y_m, X_1,\ldots,X_m\Big) \to 0 \enskip ,\]
which, by \cite{ghosal2006posterior}, happens just in case the following four conditions hold:

\begin{itemize}
\item[(P)] For every $x \in \mathcal{X}$, the covariance function $\sigma_0(x,\cdot)$ has continuous partial derivatives up to order $2\alpha + 2$, the mean function $\mu(x) \in \bar{\mathcal{A}}$, the RKHS of $\sigma_0(\cdot, \cdot)$, and the prior $\Pi_\lambda$ is fully supported on $(0, \infty)$ \enskip ;
\item[(C)] The covariate space $\mathcal{X}$ is a bounded subset of $\mathbb{R}^d$ \enskip ;
\item[(T)] The transformed true response function $\eta_0 \in \bar{\mathcal{A}}$ \enskip ;
\item[(G)] For every $b_1 > 0$ and $b_2 > 0$, there exist sequences $M_n, \tau_n$ and $\lambda_n$ such that
\[M_m^2\tau_m\lambda_m^{-2} \geq b_1m \qquad \text{ and } \qquad M_m^{d/\alpha} \leq b_2m \enskip .\]
\end{itemize}

It is straightforward to verify conditions $(P)$, $(C)$, and $(G)$:

\begin{itemize}
\item[(P)] The squared-exponential kernel is infinitely divisible, so we can take any $\alpha \in \mathbb{N}$ in $(P)$.
Furthermore, we are taking $\mu(x) = 0 \in \bar{\mathcal{A}}$ since all RKHS contain the identity element.
Finally, we put an inverse-gamma prior on the bandwidth $\ell$, which is fully supported on the positive reals;
\item[(C)] The boundedness of our covariate space follows from the fact that we can embed the space of weighted networks in $\mathbb{R}^{n^2}$, where $n$ is the number of nodes;
\item[(G)] We can take $M_m = \ell_m$ and $\tau_m = b_1m$.
Then the result holds, since we can take any $\alpha \in \mathbb{N}$ \enskip .
\end{itemize}
Note that the embedding of the space of weighted networks in $\mathbb{R}^{n^2}$ is given by \cite{ginestet2017hypothesis}.

Our only remaining obstacle to showing posterior consistency is $(T)$.
\cite{tokdar2007posterior} provide an equivalency for $(T)$.
The authors show that
\[\eta_0 \in \bar{\mathcal{A}} \iff \forall \varepsilon > 0, \ \mathbb{P}(\Vert \eta(x) - \eta_0\Vert_\infty < \varepsilon) > 0\]
if the following three conditions hold:

\begin{itemize}
\item[(A1)] $\exists M, M' > 0$ such that $M' \leq \sigma_0(t,t) \leq M, \ \forall t \in (\mathbb{R}^+)^d$ \enskip ;
\item[(A2)] $\exists C, q > 0$ such that $[\sigma_0(t,t) + \sigma_0(s,s) - 2\sigma_0(t,s)]^{1/2} \leq C\Vert s-t \Vert^q, \ \forall s, t \in (\mathbb{R}^+)^d$ \enskip ;
\item[(A3)] For any $n\geq 1$ and any $t_1, \ldots,t_n \in (\mathbb{R}^+)^d, \Sigma = ((\sigma(t_i, t_j)))$ is nonsingular.
\end{itemize}

First, we verify these conditions.

\begin{itemize}
\item[(A1)] Again, since we are using the squared-exponential kernel, for any distance $d(\cdot, \cdot)$, we have $d(t, t) = 0$ so that $\exp(-d^2(t,t)/2\sigma^2) = 1$.
Therefore, take $M' = 1 = M$ \enskip ;
\item[(A2)] Since $\sigma_0(t,s) \leq 1$, we can take $q=1$ to obtain
\[[1+1-2\sigma_0(t,s)]^{1/2} \leq \sqrt{2} \enskip .\]
Therefore, take $C=\sqrt{2}/\max_{s,t\in(\mathbb{R}^+)^d}\vert s - t|$ \enskip ;
\item[(A3)] We already proved that a squared-exponential kernel induced by the Frobenius distance in \eqref{dist:Frobenius} or the spectral-based distance in \eqref{dist:lambda} is positive-definite.
\end{itemize}

So, we are done if we can show
\[\forall \varepsilon > 0, \ \mathbb{P}(\Vert \eta(x) - \eta_0\Vert_\infty < \varepsilon) > 0 \enskip .\]
For $\varepsilon > 0$, we have

\begin{align*}
\mathbb{P}(\Vert \eta(x) - \eta_0\Vert_\infty < \varepsilon) &\geq
\mathbb{P}(\Vert \eta(x)\Vert_\infty + \Vert\eta_0\Vert_\infty < \varepsilon) = \mathbb{P}(\Vert \eta(x)\Vert_\infty < \varepsilon - |\eta_0|) \\
&=\mathbb{P}(\sup_x\vert \eta(x) \vert < \varepsilon - |\eta_0|) = 1 - \mathbb{P}(\sup_x\vert \eta(x)\vert > \varepsilon - |\eta_0|) \\
&\geq 1 - \exp(-(\varepsilon-|\eta_0|)^2 / 2) > 0 \enskip ,
\end{align*}
where we used the Borell, or Borell-TIS, inequality in the fifth step, which says that for a mean-zero GP, $X$, with $\sigma^2(X)) = \sup\text{Var}(X)$, we have
\[\mathbb{P}(\Vert X\Vert_\infty \geq x) \leq 2\exp(-x^2/2\sigma^2(X)) \enskip .\]
Hence, $\eta_0 \in \bar{\mathcal{A}}$, verifying condition $(T)$ and concluding our proof.
\end{proof}

\begin{proof}[Proof of Theorem~\ref{thm: surv cons}] \label{pf: surv cons}
\cite{fernandez2016posterior} show that for a randomized design, the model for survival analysis in Section~\ref{Survival Analysis} achieves posterior consistency if the following four conditions hold:

\begin{itemize}
\item[(A1)] $\big(k(0) - k(2^{-n})\big)^{-1} \geq n^6$ \enskip ;
\item[(A2)] $\nu$ assigns positive probability to every neighborhood of the true parameter;
\item[(A3)] $\mathbb{E}_{\theta_0}[T] < M < \infty$ \enskip ;
\item[(A4)] The true parameters $\eta_{j,0}$ take the form $\hat{\eta}_{i, 0}/h_d$, where $\hat{\eta}_{j,0}$ is in the support of $\hat{\eta}_i$ under the uniform norm.
$\hat{\eta}_j = h_d\eta_j$ is the non-stationary GP as a function of the original GP with
\[h_d(t) = \begin{cases} \frac{d+1}{1+\log(1-e^{-1})} & t < 1 \\ \frac{d+1}{t+\log(1-e^{-t})} & t \geq 1 \enskip . \end{cases}\]
\end{itemize}

The authors state that these assumptions are ``quite reasonable and natural for the type of data we were dealing with."
In particular, we assume (A3) and (A4) hold.
In words, (A3) just says that survival times have finite expectation, which is true for survival times like death.
For others, like diagnosis, it is ill-posed to consider whether the likelihood goes to one as time increases.
In general, this is a very reasonable assumption.
On the other hand, (A4) is unverifiable, so we take it on faith.

For (A1), we note that for squared-exponential kernels, which are stationary, we have
\[k(s) = \sigma^2 \cdot \exp\Big(-\frac{s^2}{2\ell^2}\Big) = \sigma^2 \cdot \exp\Big(-\tilde{\ell}\cdot s^2\Big) \enskip ,\]
for some distance $s$.
Hence, there is no ambiguity about what stationarity means for graph kernels, as the stationarity is with respect to, in our case, both the Frobenius distance and the spectral-based distance.
Therefore, we have
\[\big(k(0) - k(2^{-n})\big)^{-1} \geq n^6 \iff \frac{\big(1 - \exp(-\frac{\ell}{4^n})\big)^{-1}}{\sigma^2 \cdot n^6} \geq 1 \enskip .\]
This is true asymptotically, as the limit of the left side of the inequality goes to $\infty$ in $n$.

For (A2) We take a constant baseline hazard function $\omega = 2\cdot\Omega$.
This corresponds to a hazard function of an exponential random variable with mean $1/\Omega$, i.e. $\omega_0 \sim \nu = \text{Exp}(\Omega)$, which is fully supported on the positive reals.
\end{proof}

\end{appendix}

\end{document}